\documentclass[aps,groupedaddress,superscriptaddress,amssymb,amsmath,amsbsy,amsfonts,twocolumn,pra]{revtex4}
\usepackage{amsmath}
\usepackage{amssymb}
\usepackage{graphicx,color}
\usepackage{amsthm}
\usepackage{hyperref}
\usepackage{extarrows}
\usepackage{adjustbox,lipsum}
\usepackage{makecell}
\usepackage{bbold}

\begin{document}

\newcommand*{\cl}[1]{{\mathcal{#1}}}
\newcommand*{\bb}[1]{{\mathbb{#1}}}
\newcommand{\ket}[1]{|#1\rangle}
\newcommand{\bra}[1]{\langle#1|}
\newcommand{\inn}[2]{\langle#1|#2\rangle}
\newcommand{\proj}[2]{| #1 \rangle\!\langle #2 |}
\newcommand*{\tn}[1]{{\textnormal{#1}}}
\newcommand*{\1}{{\mathbb{1}}}
\newcommand{\T}{\mbox{$\textnormal{Tr}$}}
\newcommand*{\todo}[1]{\textcolor[rgb]{0.99,0.1,0.3}{#1}}

\theoremstyle{plain}
\newtheorem{prop}{Proposition}
\newtheorem{proposition}{Proposition}
\newtheorem{theorem}{Theorem}
\newtheorem{lemma}[theorem]{Lemma}
\newtheorem{remark}{Remark}

\theoremstyle{definition}
\newtheorem{definition}{Definition}

\title{Estimating Quantum Mutual Information Through a Quantum Neural Network}
\author{Myeongjin Shin}
\email{hanwoolmj@kaist.ac.kr}
\affiliation{School of Computing, Korea Advanced Institute of Science and Technology (KAIST), Daejeon 34141, Korea}

\author{Junseo Lee}
\email{js\_lee@norma.co.kr }
\affiliation{School of Electrical and Electronic Engineering, Yonsei University, Seoul 03722, Korea}
\affiliation{Quantum Security R\&D, Norma Inc., Seoul 04799, Korea}

\author{Kabgyun Jeong}
\email{kgjeong6@snu.ac.kr}
\affiliation{Research Institute of Mathematics, Seoul National University, Seoul 08826, Korea}
\affiliation{School of Computational Sciences, Korea Institute for Advanced Study, Seoul 02455, Korea}

\date{\today}

\begin{abstract}
We propose a method of quantum machine learning called quantum mutual information neural estimation (QMINE) for estimating von Neumann entropy and quantum mutual information, which are fundamental properties in quantum information theory. The QMINE proposed here basically utilizes a technique of quantum neural networks (QNNs), to minimize a loss function that determines the von Neumann entropy, and thus quantum mutual information, which is believed more powerful to process quantum datasets than conventional neural networks due to quantum superposition and entanglement. To create a precise loss function, we propose a quantum Donsker-Varadhan representation (QDVR), which is a quantum analog of the classical Donsker-Varadhan representation. By exploiting a parameter shift rule on parameterized quantum circuits, we can efficiently implement and optimize the QNN and estimate the quantum entropies using the QMINE technique. Furthermore, numerical observations support our predictions of QDVR and demonstrate the good performance of QMINE.
\end{abstract}

\maketitle

\section{\label{sec:Introduction}Introduction}
The concept of quantum mutual information (QMI) in quantum information theory quantifies the amount of information shared between two quantum systems. This extends the classical notion of mutual information to the quantum regime~\cite{J07,NC00,W17}. This information measure is fundamental, because it determines the quantum correlation or entanglement between two quantum systems. The information obtained from quantum mutual information can be applied to various fields of quantum information processing such as quantum computation, quantum cryptography, and quantum communication~\cite{NC00,W17} (particularly in quantum channel capacity problems~\cite{BS04,H20}). They also play a crucial role in quantum machine learning~\cite{BWP+17,CCC+19}, where they measure the information shared between different representations of quantum datasets. Moreover, the gathered information can be used to enhance the efficiency and effectiveness of quantum algorithms in processing quantum data.

Quantum mutual information is expressed as the sum of von Neumann entropies, denoted by $S(\rho)=-\tn{Tr}(\rho\ln\rho)$ for a quantum state $\rho$, making the determination of the von Neumann entropy~\cite{BZ17} essential for calculating quantum mutual information. In recent years, the estimation of the von Neumann entropy has garnered significant attention in the field of quantum information theory. Various methods have been proposed to estimate von Neumann entropy, including those exploiting quantum state tomography~\cite{OW16}, Monte Carlo sampling~\cite{HGKM10}, and entanglement entropy~\cite{CCD09,GHS21,AISW20,TV21,WZW22,WGL+22,GL19,SH21}. Several studies~\cite{WGL+22,GL19,SH21} have utilized quantum query models for entropy estimation and have demonstrated promising quantum speedups. Specifically, Wang \emph{et al}.~\cite{WGL+22} proposed that the von Neumann entropy can be estimated with an accuracy of $\varepsilon$ by using $O(\frac{r^2}{\varepsilon^2})$ queries. However, these query model-based algorithms have practical limitations because a quantum circuit that generates the quantum state must be prepared, and the effectiveness of constructing a quantum query model for the input state remains an open question~\cite{WZW22}. Thus, we focused on estimating the von Neumann entropy of an unknown quantum state using only identical copies of the state. To the best of our knowledge, no existing quantum algorithms estimate the von Neumann entropy using $O(\text{poly}(r), \text{poly}(\frac{1}{\varepsilon}))$ copies of the quantum state, where $r$ represents the rank of the state.

A mutual information neural estimation (MINE) method is a novel technique that utilizes neural networks to calculate the classical mutual information between two random variables. More precisely, this method optimizes a neural network to estimate mutual information by minimizing the loss function. The loss function is based on the Donsker-Varadhan representation~\cite{vN96} that provides a lower bound for the well-known Kullback-Leibler (KL) divergence.

Quantum neural networks (QNNs)~\cite{BBF+20,WDK+17}, which are among the most powerful quantum machine learning methods, serve as quantum counterparts to conventional neural networks, and offer several advantages. One notable advantage is the ability to use a quantum state as an input, which is particularly advantageous when calculating quantum mutual information or the von Neumann entropy. We identified two types of QNNs in the literature~\cite{BBF+20,BLSF19} that possess a neural network structure and leverage quantum advantages accompanied by well-defined training procedures. In this study, we employed a parameterized quantum circuit~\cite{BLSF19}, which is known for its quantum advantages, despite the presence of the barren plateau problem, which requires further investigation~\cite{MBS+18}.

As a quantum analog of MINE, we propose a quantum mutual information neural estimation (QMINE), which is a method for determining the von Neumann entropy and quantum mutual information through a quantum neural network technique. Similar to the classical case, QMINE uses a quantum neural network to minimize the loss function that evaluates the von Neumann entropy. To generate a loss function that estimates the von Neumann entropy, we present the quantum Donsker-Varadhan representation (QDVR), which is a quantum version of the Donsker-Varadhan representation. QMINE offers the potential for a quantum advantage in estimating the von Neumann entropy facilitated by QDVR. By converting the problem of von Neumann entropy estimation into a quantum machine learning regime, QMINE opens new possibilities. There is also the potential to estimate von Neumann entropy using only $O(\text{poly}(r), \text{poly}(\frac{1}{\varepsilon}))$ copies of the quantum state. However, we acknowledge that further investigation is required owing to the challenging and well-known barren plateau problem, as well as the need for efficient quantum training methods.

The remainder of this paper is organized as follows. In Sec.~\ref{sec:background}, we briefly introduce the basic notions of quantum mutual information, MINE, and parame-
trized quantum circuits. In Sec.~\ref{sec:QDV}, we generalize the Donsker-Varadhan representation to a QDVR, which is the main component of QMINE. We also propose an estimation method for von Neumann entropy using quantum neural networks in Sec. ~\ref{sec:MAIN}. This implies that it is possible to efficiently obtain the quantum mutual information, and its numerical simulations under the framework of QMINE in Sec. ~\ref{sec:ns}. Finally, a discussion and remarks are presented in Sec. ~\ref{sec:conclusion}, and open questions and possibilities are raised for future research.

\textbf{Note on concurrent work.} The independent and concurrent work~\cite{GPSW23} appeared on the arXiv a few days after our preprint was uploaded. It introduced a method for estimating von Neumann entropy reminiscent of ours, with R\'enyi entropy, measured relative (R\'enyi) entropy, and fidelity. Our work focused on estimating von Neumann entropy with low copy complexity. We reduced the domain in the variation formula but Ref.~\cite{GPSW23} did not. We believe that limiting the trace and rank in the variation formula is crucial for effective estimation.

\section{Preliminaries} \label{sec:background}
\subsection{Quantum Mutual Information and von Neumann Entropy}
Quantum mutual information, also known as von Neumann mutual information, quantifies the relationship between two quantum states. This can be calculated by using the formula (See Fig. ~\ref{Fig1}):
\begin{align}
I\left(A:B\right) = S\left(\rho^A\right) + S\left(\rho^B\right) - S\left(\rho^{AB}\right).
\end{align}
Here, $S(\rho)$ represents the von Neumann entropy \cite{BZ17} of quantum state $\rho$ in a $d$-dimensional Hilbert space, given by $S\left(\rho\right) = -\text{Tr}\left(\rho \log\rho\right)$. Therefore, estimating the von Neumann entropy enables the estimation of the quantum mutual information.

\begin{figure}
\centering
\includegraphics[width=0.8\linewidth]{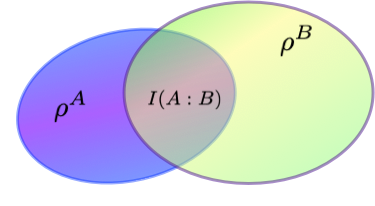}
\caption{Schematic diagram for the quantum mutual information, $I(A:B)$, between two quantum states $\rho^A$ and $\rho^B$.}
\label{Fig1}
\end{figure}

The von Neumann entropy, which is an extension of the Shannon entropy~\cite{S48} to the quantum domain, can be estimated using quantum circuits and measurements. It is defined as the entropy of the density matrix associated with a quantum state, where the density matrix is a positive semi-definite matrix that represents the state. To estimate the von Neumann entropy, measurements can be performed on multiple copies of the quantum state and the outcomes of these measurements can be utilized. The most straightforward approach is to directly estimate the density matrix and calculate the entropy using its definition. However, estimating the von Neumann entropy can be challenging, particularly for large quantum systems, owing to the difficulty in accurately estimating the density matrix. Furthermore, the estimation accuracy is influenced by the number of measurements conducted and the quality of the quantum hardware employed. However, ongoing research is focused on developing more efficient and precise methods for estimating the von Neumann entropy.

Several methods have been employed to estimate von Neumann entropy, particularly those utilizing the quantum query model \cite{WZW22,GL19,SH21}. In the quantum query model, if the quantum circuit $U$ produces a quantum state $\rho$, it utilizes unitary gates such as $U$, $U^{\dagger}$, and C$U$ (controlled-$U$). However, the quantum circuit must be known to use the query model. The effectiveness of constructing a quantum query model for a given input state remains uncertain~\cite{WZW22}, prompting us to explore the von Neumann entropy estimation without relying on the quantum query model. In the absence of a query model, our approach solely exploits identical copies of quantum states. Previous studies, such as Acharya \emph{et al}.~\cite{AISW20} employed $O(d^2)$ copies of the quantum state $\rho$, where $d$ denotes the dimension, whereas Wang \emph{et al}.~\cite{WZW22} used $O\left(\frac{1}{\varepsilon^{5}\lambda^{2}}\right)$ copies of $\rho$, where $\lambda$ represents the lower bound on all nonzero eigenvalues. To the best of our knowledge, no existing algorithm provides a high-accuracy estimation of the von Neumann entropy by using only $O(\text{poly}(r))$ copies of $\rho$ with rank $r$.

\subsection{Mutual Information Neural Estimator}
The mutual information neural estimator (MINE)~\cite{BBR+18} is a method for estimating the mutual information of two random variables by using neural networks. This approach involves selecting functions $T_{\theta}: X \times Y \rightarrow \mathbb{R}$ that are parameterized by neural networks with the parameter $\theta \in \Theta$. Considering $n$ samples, we define the empirical joint and product probability distributions as $p_{XY}^{(n)}$ and $p_X^{(n)} \times p_Y^{(n)}$, respectively. The MINE strategy is given by:

\begin{equation} \label{eq:nim}
\widehat{I(X:Y)_n} = \sup_{\theta \in \Theta} \mathbb{E}_{p_{XY}}\left[T_{\theta}\right]-\log\left(\mathbb{E}_{p_X \times p_Y}\left[e^{T_{\theta}}\right]\right),
\end{equation} 
where $\bb{E}$ is the expected value. Additionally, the Donsker-Varadhan representation is defined as follows: For any probability distribution functions $p$ and $q$,

\begin{align}
D_{KL}(p||q) = \sup_{T: \Omega \rightarrow \mathbb{R}} \mathbb{E}_p[T]-\log\left(\mathbb{E}_q\left[e^T\right]\right),
\end{align}
where we take the supremum over all the functions $T$.

Using the Donsker-Varadhan representation~\cite{DV76}, it can be proven that $I\left(X:Y\right) \geq \widehat{I\left(X:Y\right)_n}$ and MINE are strongly consistent, meaning that there exists a positive integer $N$ and a choice of neural network parameters $\theta \in \Theta$ such that for all $n \geq N$, $\left|I(X:Y) - \widehat{I(X:Y)_n}\right| \leq \varepsilon$. By applying a gradient descent method on the neural network $T_{\theta}$ to maximize $\mathbb{E}_{p_{XY}}[T_{\theta}]-\log\left(\mathbb{E}_{p_X \times p_Y}\left[e^{T_{\theta}}\right]\right)$, we can obtain $\widehat{I(X:Y)_n}$ and estimate the mutual information $I(X:Y)$.

The MINE technique has found applications in various areas of artificial intelligence, such as feature selection, representation learning, and unsupervised learning, using information-theoretic methods. Compared to previous approaches, it provides more accurate and robust estimates of mutual information, leading to significant advancements in the field of artificial intelligence (AI). It is important to recognize that MINE is a relatively new and rapidly evolving field, with ongoing research focused on enhancing and broadening its capabilities. Nonetheless, the MINE technique is widely regarded as a valuable tool in AI and information theory, offering a powerful and flexible approach for estimating the mutual information between variables.

\subsection{Parametrized Quantum Circuits}
Parameterized quantum circuits (PQCs)~\cite{BLSF19} are quantum circuits that incorporate adjustable parameters, typically represented as real numbers. These parameters can be fine-tuned to control the behavior of the quantum circuit, thereby providing increased flexibility and optimization potential. Parameterized quantum circuits have extensive applications in quantum machine learning and optimization algorithms, enabling computations that are challenging or even infeasible using classical methods. The key concept is to employ a parameterized quantum circuit as a feature extractor or waveform generator, followed by classical optimization algorithms that iteratively adjust the circuit parameters to minimize the objective function.

By manipulating circuit parameters, one can efficiently learn and represent quantum systems in a compact and adaptable manner. In quantum optimization, parameterized quantum circuits play a crucial role in global minimum search. By encoding the objective function into circuit parameters, quantum effects such as quantum parallelism and quantum tunneling can be harnessed to explore the search space more effectively than classical optimization algorithms.
The objective function can be represented as a measurement outcome of the quantum circuit. The quantum circuit can harness superposition and entanglement to explore the search space more effectively than classical optimization algorithms.

One of the core techniques used in quantum optimization procedures for parameterized quantum circuits is the parameter shift rule~\cite{MNKF18}. The parameter shift rule is a powerful tool in quantum machine learning that enables efficient computation of gradients with respect to the parameters of a quantum circuit. 

The fundamental concept behind the parameter shift rule is to employ a quantum circuit with adjustable parameters to perform the measurements. By utilizing the measurement outcome, it is possible to estimate the gradient of a cost function with respect to the circuit parameters. This rule capitalizes on the notion that small variations in the parameters of a quantum circuit can be used to calculate the derivative of the cost function pertaining to these parameters.

The underlying principle involves the preparation of two identical copies of a quantum state, each with slightly different parameter values. By comparing these two quantum states, it was possible to estimate the gradient. More importantly, this method allows the calculation of gradients in a single pass through a quantum circuit, obviating the need for additional measurements. By using multiple samples via measurement, the gradient can be estimated.

If a parameterized quantum circuit is represented as a sequence of unitary gates, it is denoted as 
\begin{equation*}
U(x;\theta):= \left(\prod_{i=1}^{N} U_i(\theta_i)\right) U_0(x).
\end{equation*}
The output of the circuit can then be observed using an observable $\hat{O}$ and the measurement outcome becomes a quantum circuit function. The quantum circuit function is expressed in simplified form as $f(x;\theta_i) = \langle\psi_i| U_i^\dagger(\theta_i) \hat{O}_{i+1} U_i(\theta_i) |\psi_i\rangle$ for each $i$. The gradient of the quantum circuit function can then be calculated using the parameter shift rule, as follows:

\begin{widetext}
\begin{equation} \label{eq:psr}
\nabla_{\theta_i} f(x;\theta_i) = c \left(\langle\psi_i| U_i^\dagger(\theta_i+s) \hat{O}_{i+1} U_i(\theta_i+s) |\psi_i\rangle - \langle\psi_i| U_i^\dagger(\theta_i-s) \hat{O}_{i+1} U_i(\theta_i-s) |\psi_i\rangle\right).
\end{equation}
\end{widetext}

The parameter shift rule has been successfully employed in various quantum machine learning algorithms, including quantum neural networks~\cite{BBF+20,BLSF19} and quantum support vector machines~\cite{RML14,HCT+19}, for optimization and training purposes. It is regarded as a valuable tool for developing efficient quantum machine learning algorithms, as it enables the efficient computation of gradients in quantum systems, which is often a challenging task. It is important to note that the parameter shift rule is an approximation, and its accuracy depends on factors such as the choice of parameters, cost function, and the specific quantum circuit. Nevertheless, it has proven to be a useful and efficient technique in the emerging field of quantum machine learning, and our ongoing research focuses on enhancing and expanding its potential capabilities.

\section{Quantum Donsker-Varadhan Representation} \label{sec:QDV}
The quantum Donsker-Varadhan representation is a mathematical framework that enables quantum neural networks to estimate the von Neumann entropy. It is a quantum counterpart of the original Donsker-Varadhan representation, with the distinction that QDVR focuses solely on the quantum entropy rather than on the relative entropy. QDVR can be considered as a modified version of the Gibbs variational principle~\cite{H68}, which restricts the domain to density matrices.

As mentioned previously, MINE~\cite{BBR+18} exploits the original Donsker-Varadhan representation to estimate classical mutual information using a (classical) neural network. In the context of estimating mutual quantum information, it is natural to consider a quantum version of the Donsker-Varadhan representation. Notably, we need only estimate the components of von Neumann entropy $S(\rho_A)$, $S(\rho_B)$, and $S(\rho_{AB})$ to determine the quantum mutual information $I(A:B)$. A variational formula for von Neumann entropy exists as follows:

\begin{theorem}[Gibbs Variational Principle~\cite{H68}] \label{gibbs}
Let $f: H^{d \times d} \rightarrow \mathbb{R}$ be a function defined on $d$-dimensional Hermitian matrices $T$ and $\rho$ be a density matrix. Then we have
\begin{equation} \label{eq:qdv}
f(T) = -\T(\rho T) + \log\left(\T(e^T)\right).
\end{equation}
Thus, for $d$-dimensional Hermitian matrices $T$, the von Neumann entropy is given by:
\begin{equation} \label{identity}
S(\rho) = \inf_{T} f(T),
\end{equation}
where the infimum is taken over all Hermitian $T$.
\end{theorem}

Our objective is to determine the Hermitian matrix $T$ that maximizes $f(T)$. We parameterize $T$ by using $t_i \in \mathbb{R}$ and $\ket{\psi_i} \in \mathbb{C}^d$. We can express $T = \sum_{i=1}^r t_i \ket{\psi_i}\bra{\psi_i}$, which gives us $f(T) = -\sum_{i=1}^d t_i \bra{\psi_i}\rho\ket{\psi_i} + \log\left(\sum_{i=1}^d e^{t_i}\right)$. To compute $f(T)$, we must measure the quantum state $\rho$ using the basis $\{\ket{\psi_i}\}_{i=1}^d$. Achieving this with an error smaller than $\varepsilon$ requires $O(\frac{\sigma^2}{\varepsilon^{2}})$ samples of $\rho$, where $\sigma:=\text{Var}(\{t_i\})$. However, the number of required samples of $\rho$ can become substantial because of the broad domain of $T$, which encompasses all Hermitian matrices. Therefore, reducing the size of this domain is imperative.

\begin{lemma} \label{lem2}
For all Hermitian matrices $T$, the function $f$ holds that
\begin{equation} \label{identity}
f(T) = f(T+cI) 
\end{equation}
for a constant $c$.
\end{lemma}

\begin{proof}
For any $T \in H^{d \times d}$, we have
\begin{align*}
f(T+cI) &= -\text{Tr}(\rho(T+cI)) + \log\left(\text{Tr}(e^{T+cI})\right) \\
&= -\text{Tr}(\rho T) - c\text{Tr}(\rho) + \log\left(e^c\text{Tr}(e^T)\right) \\
&= -\text{Tr}(\rho T) - \log\left(\text{Tr}(e^T)\right) \\
&= f(T).
\end{align*}
Thus, $f(T) = f(T+cI)$ for a constant $c$.
\end{proof}

\begin{proposition}[Domain Reduction] \label{}
Let $f: H^{d \times d} \rightarrow \mathbb{R}$ be a function defined on $d$-dimensional Hermitian matrices and let $\rho$ be a density matrix. Then,
\begin{equation} \label{identity}
S(\rho) = \inf_Tf(T)
\end{equation}
for $d$-dimensional `positive' Hermitian matrices $T$.
\end{proposition}

\begin{proof}
For any Hermitian matrix $T \in H^{d \times d}$, let $c = \max_{\ket{\psi_i}\in \mathbb{C}^d}(-\bra{\psi_i}T\ket{\psi_i})$. From Lemma~\ref{lem2}, there exists a \emph{positive} Hermitian matrix $T_0 = T+cI$ such that $f(T)=f(T_0)$. Therefore, we can reduce the domain of $T$ to a positive Hermitian matrix.
\end{proof}

Now, we only need to search for the space of the positive Hermitian matrices to find the optimal $T$. The computational complexity of copying $\rho$ to calculate $f(T)$ depends on $T$. To reduce this complexity, we need to specify and limit the trace of $T$.

\begin{lemma} \label{lem3}
A positive Hermitian matrix $T_0$ with rank $r$ exists that satisfies $\T(T_0) \leq 2rn + r\log\left(\frac{1}{\varepsilon}\right)$ such that
\begin{equation} \label{trace reduction}
\left|S(\rho) - f(T_0)\right| < \varepsilon,
\end{equation}
where $\rho$ is an $r$-rank density matrix.
\end{lemma}

\begin{proof}
See the details of the proof in Appendix~\ref{app:a}.
\end{proof}

\begin{proposition}[Quantum Donsker-Varadhan Representation] \label{prop:qdvr}
Let $f: H^{d \times d} \rightarrow \mathbb{R}$ be a function defined on $d$-dimensional Hermitian matrices, and let $\rho$ be an $r$-rank density matrix.
\begin{equation} \label{eq:qdv}
g(T) =-\T\left(c \rho T\right) + \log\left(\T(e^{cT})\right),
\end{equation}
where $c \geq 2rn + r\log\left(\frac{1}{\varepsilon}\right)$. Then,
\begin{equation} \label{qdv}
\left|S(\rho) - \inf(g(T))\right| < \varepsilon
\end{equation}
for any $d$-dimensional $r$-rank density matrix $T$.
\end{proposition}

\begin{proof}
By using Lemma~\ref{lem2} and Lemma~\ref{lem3}, there exists an $r$-rank positive Hermitian matrix $T_0$ with $\text{Tr}(T_0) = c$ such that $|S(\rho) - f(T_0)| < \varepsilon$. Thus, $T_1 = \frac{T_0}{c}$ is an $r$-rank density matrix, and $|S(\rho) - g(T_1)| < \varepsilon$. From Theorem~\ref{gibbs}, $S(\rho) \leq g(T)$ for all density matrices $T$. Therefore, $\left|S(\rho) - \inf(g(T))\right| \leq \left|S(\rho) - g(T_1)\right| < \varepsilon$. This completes the proof.
\end{proof}

According to the quantum Donsker-Varadhan representation in Proposition~\ref{prop:qdvr}, we only need to search within the space of the density matrices. By calculating $g(T)$ with an error of $\varepsilon$, the complexity of copying $\rho$ is $O(\frac{c^2}{\varepsilon^2})$. Next, we plan to determine the optimal density matrix $T$ that minimizes $g(T)$. In the next section, we will use quantum neural networks to determine the optimal $T$.
\bigskip

\section{Von Neumann Entropy Estimation with Quantum Neural Networks} \label{sec:MAIN}
We now explain the estimation of von Neumann entropy using quantum neural networks, specifically focusing on parameterized quantum circuits as an example. Our approach is inspired by the work of Liu \emph{et al}.~\cite{LMZW21}, who utilized variational autoregressive networks and quantum circuits to address problems in quantum statistical mechanics. To achieve this, specific values are assigned to the variables in $T$ by defining $t$ as a set of real numbers, $\{t_i | t_i \in \mathbb{R}\}$ and $\ket{\psi_i}$ as complex vectors in $\mathbb{C}^d$. Additionally, let us assume that the rank of $\rho$ is denoted by $r$, and we define $T = \sum^r_{i=1} t_i \ket{\psi_i}\bra{\psi_i}$.

Consequently, the function $g(T)$ becomes $g(T) = -c\sum^r_{i=1} t_i \bra{\psi_i}\rho\ket{\psi_i} + \log\left(d-r+\sum^r_{i=1} e^{ct_i}\right)$. We can introduce a unitary operator $U$ that transforms $\ket{\psi_i}$ into $\ket{i}$, and represent this unitary operator using a set of parameters $\theta$ as $U(\theta)$ as follows:

\begin{equation}
g(T) = -c\sum^r_{i=1} t_i \bra{i}U(\theta)\rho U^{\dagger}(\theta)\ket{i} + \log\left(d-r+\sum^r_{i=1} e^{ct_i}\right).
\end{equation}

By considering $U(\theta)$ as a quantum neural network and $\rho$ as its input, we can obtain the network output by computing $U(\theta)\rho U^{\dagger}(\theta)$. To accurately calculate $g(T)$ with an error rate less than $\varepsilon$, it is necessary to measure the output of the quantum neural network $O\left(\frac{\text{Var}({ct_i})^2}{\varepsilon^{2}}\right) = O\left(\frac{c^2}{\varepsilon^{2}}\right)$ times.

Our objective was to optimize the parameters to determine the infimum of $g(T)$. For example, let us consider a parameterized quantum circuit~\cite{BLSF19} with Pauli gates as a quantum neural network $U(\theta) = \prod^k_{i=1} U(\theta_i)$, where $U(\theta_i) = e^{-i\frac{\theta_i}{2}P_i}$. By applying the parameter shift rule~\cite{MNKF18}, we observe that

\begin{equation}
\nabla_{\theta}g(t, \theta) = \frac{1}{2}\left[g\left(t, \theta+\frac{\pi}{2}\right)-g\left(t, \theta-\frac{\pi}{2}\right)\right],
\end{equation}
and

\begin{equation}
\frac{\partial g(t, \theta)}{\partial t_i} = -c\bra{i}U(\theta)\rho U^{\dagger}(\theta)\ket{i} + \frac{ce^{ct_i}}{d-r+\sum^r_{i=1} e^{ct_i}}.
\end{equation}

To satisfy the conditions $t_i \geq 0$ and $\sum^r_{i=1} t_i = 1$, we choose $t_i = \left(\prod^{i-1}_{j=1} \sin^2\varphi_j\right)\left(\cos^2\varphi_i\right)$. We can apply gradient descent to $\varphi_j$ and $\theta_i$ to optimize the quantum circuit. To calculate the gradient, we require $O\left(\frac{c^2}{\varepsilon^{2}} \times \left(\text{\# of parameters in QNN}\right)\right)$ copies of $\rho$. Therefore, to obtain $\inf\left(g\left(T\right)\right)$ and estimate $S\left(\rho\right)$ with an error of less than $\varepsilon$, we require

\begin{widetext}
\begin{align}
O&\left(\frac{c^2}{\varepsilon^{2}} \times \left(\text{\# of parameters in QNN}\right) \times \left(\text{\# of trainings in QNN}\right)\right) \nonumber\\
&= O\left(\frac{r^2}{\varepsilon^2}\left(n^2+\log^2\left(\frac{1}{\varepsilon}\right)\right)n_\text{params}n_\text{train}\right)
\end{align}
\end{widetext}
copies of $\rho$.

Analytic gradient measurements in convex loss functions require $O(\frac{n^3}{\varepsilon^2})$ copies of $\rho$ to converge to a solution with $O(\varepsilon)$ close to the optimum~\cite{HN21}. In general, situations that involve parameterized quantum circuits may have nonconvex loss functions, but many algorithms still utilize parameterized quantum circuits and achieve quantum speedups. We anticipate that quantum speedup can be achieved by employing parameterized quantum circuits with analytic gradient measurements in QMINE and estimating the von Neumann entropy using $O(\text{poly}(r))$ copies of $\rho$. In future research, we will investigate the relationships between $n_\text{train}$ and $n_\text{params}$, and the performance of this approach. The key point is to transform the quantum mutual information estimation problem into a quantum neural network problem.

\section{Numerical Simulations} \label{sec:ns}
We demonstrated the performance of QMINE in estimating the quantum mutual information of random density matrices through numerical simulations of a quantum circuit. Our goal is to show that QMINE can estimate quantum mutual information with low error. We also analyze the rank and trainable parameters, and conducted simulations to support the results on QDVR.

\subsection{Rank Analysis} \label{sec:ra}
Based on QDVR, we establish that if the rank of the density matrix $\rho$ is $r$, then setting the rank of the parameter matrix $T$ to $r$ is sufficient. Thus, we aim to determine the optimal $T$ that estimates the von Neumann entropy. To investigate the effect of rank, we experimented with the rank of $T$ by letting $r=\text{rank}(\rho)$ and $k=\text{rank}(T)$. In this analysis, we simulate the scenario with $N=5, D=30, r=8$, and $c \leq 80$, where $N$ is the number of qubits, $D$ is the circuit depth, $r$ is the rank of the density matrix, $c$ is calculated using QDVR (details are provided in Appendix~\ref{app:simul}). Figure~\ref{Fig2} shows that when $k \geq r$, the result of QMINE converges to the correct value, whereas when $k < r$, it converges at a high error rate. This phenomenon has also been observed in other cases. These results support the QDVR's claim that the rank of the optimal solution $T$ is $r$. Because convergence is faster when $k=r$ than when $k>r$, it is best to use QMINE with $k=r$.

\begin{figure}
\centering
\includegraphics[width=1\linewidth]{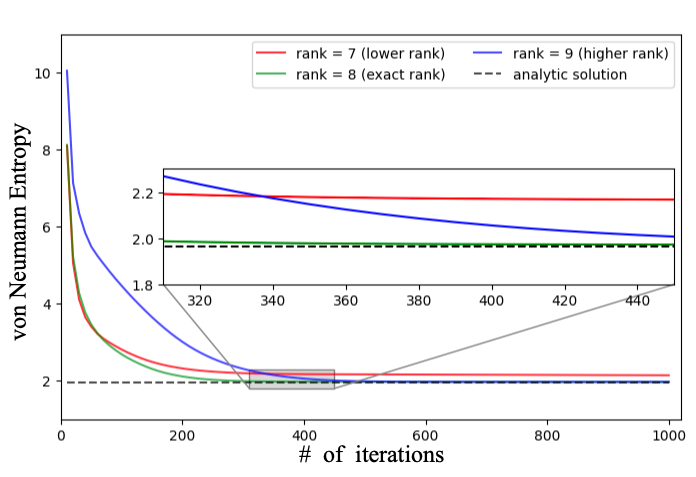}
\caption{We compare the performance of different approaches. The green curve represents QMINE with the exact rank, which exhibited the best performance. It converges rapidly with low error. However, the red curve represents QMINE with a lower rank, which converges with a high error. Finally, the blue curve represents QMINE with a higher rank, which converges with low error but at a slower pace.}
\label{Fig2}
\end{figure}

\subsection{Number of Trainable Parameters on Quantum Circuit Analysis}
We analyzed the performance of QMINE by varying the depth of the quantum circuit. In our simulations, we used $N=5$, $D=30$, $r=k=8$, and $c \leq 80$. The experimental results confirmed that as the depth of the circuit and the number of parameters increased, the estimation accuracy of QMINE improved. Fig.~\ref{Fig3} illustrates the results, showing that a circuit depth of 20 achieved the best performance. It converged rapidly with a lower error compared to a depth of 30, which converged at a slower rate despite having a similar error. These findings emphasize the importance of choosing an appropriate circuit depth (i.e., number of parameters) in QMINE. The copy complexity is determined by the number of parameters ($n_\text{params}$) and the number of training iterations ($n_\text{train}$). Therefore, when applying QMINE in various situations, it is crucial to select the correct circuit depth. We plan to investigate this aspect in future studies.

\begin{figure}
\centering
\includegraphics[width=1\linewidth]{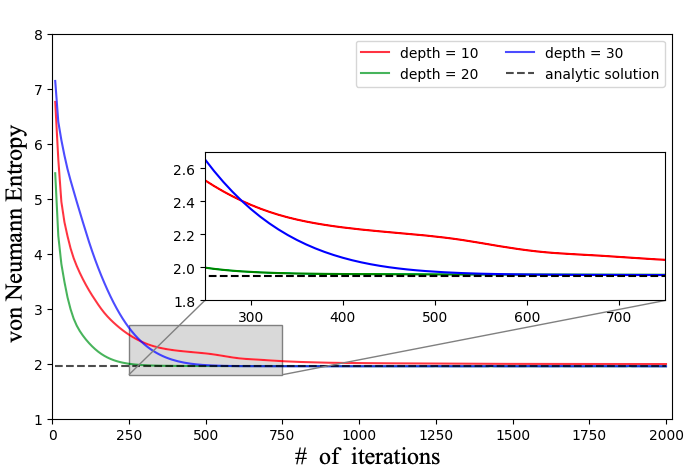}
\caption{The green line in the graph, representing a circuit depth of 20 with 400 parameters, exhibits the best performance. It converges quickly with a low error-rate. However, the red line, representing a depth of 10 with 200 parameters, converges with a high error-rate. The blue line, corresponding to a depth of 30 with 600 parameters, achieves a low error but it takes a longer-time to converge.}
\label{Fig3}
\end{figure}

\subsection{Estimating Quantum Mutual Information}
We estimated the quantum mutual information of a random density matrix using simulations with $N=4$ qubits. For each tested random density matrix, we achieved error rates ranging from $0.1\%$ to $1\%$. Additional details can be found in Appendix~\ref{app:simul}.

\begin{figure*}
\centering
\includegraphics[width=0.9\linewidth]{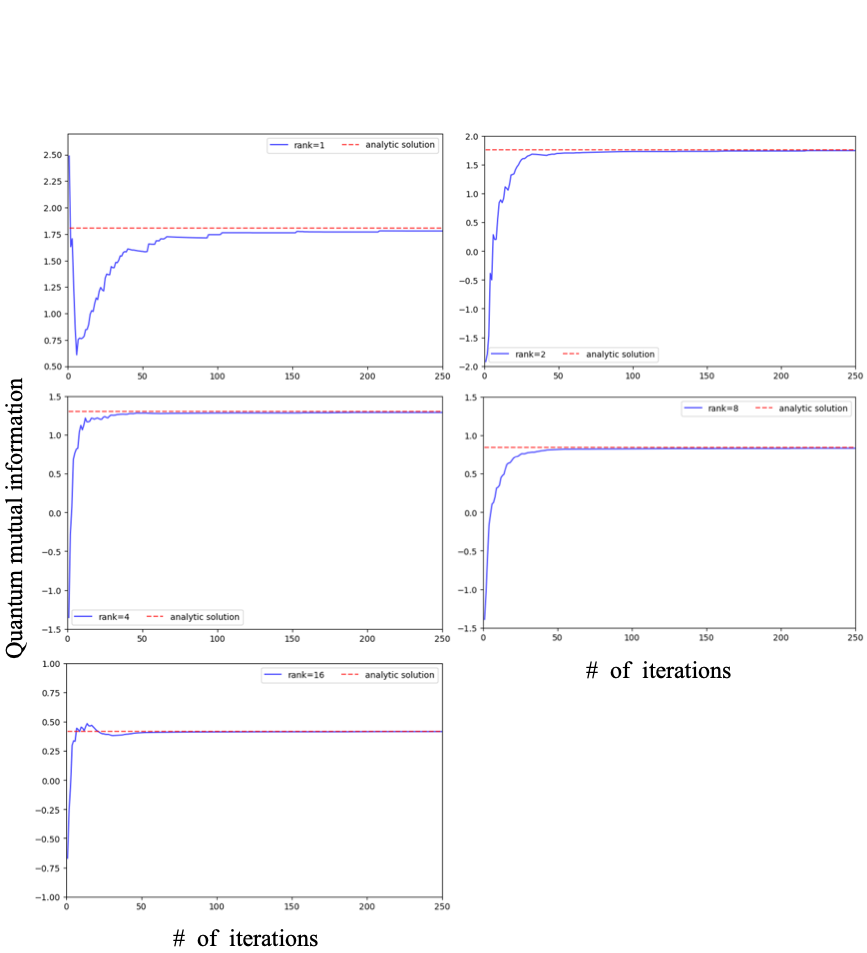}
\caption{Estimation result of the quantum mutual information of a random density matrix.}
\label{Fig4}
\end{figure*}

\section{Conclusions} \label{sec:conclusion}
We have addressed the quantum Donsker-Varadhan representation, which is a mathematical framework for estimating von Neumann entropy. The QDVR allows us to find the optimal $T$ by searching only within the density matrices, resulting in low copy complexity for calculations. By optimizing the quantum neural network using QDVR and the parameter shift rule, we can estimate von Neumann entropy and subsequently estimate the quantum mutual information. The number of copies of $\rho$ required is approximately $O\left(\frac{r^2}{\varepsilon^2}\left(n^2+\log^2\left(\frac{1}{\varepsilon}\right)\right)n_\text{params}\cdot n_\text{train}\right)$.

Through the numerical simulations, we demonstrated that the quantum mutual information neural estimation (QMINE) performs well, and it aligns with the results of quantum Donsker-Varadhan representation. The rank analysis supported the results of QDVR, whereas the circuit depth analysis emphasized the importance of selecting an appropriate circuit depth. In addition, we estimated the quantum mutual information and achieved a low error rate. The key finding of this study is the conversion of the quantum mutual information and von Neumann entropy estimation problem into a quantum neural network problem. In future, we suggest investigating the specifics of $n_\text{params}$ and $n_\text{train}$ pertaining to the quantum neural network problem. This will be explored in future studies.

\begin{acknowledgments}
This work was supported by the National Research Foundation of Korea (NRF) through grants funded by the Ministry of Science and ICT (NRF-2022M3H3A1098237) and the Ministry of Education (NRF-2021R1I1A1A01042199). This work was partially supported by an Institute for Information \& Communications Technology Promotion (IITP) grant funded by the Korean government (MSIP) (No. 2019-0-00003; Research and Development of Core Technologies for Programming, Running, Implementing, and Validating of Fault-Tolerant Quantum Computing Systems).
\end{acknowledgments}

\section*{Data availability}
Our manuscript has no associated data.

\section*{Conflict of Interest}
The authors have no conflicts to disclose.

\section{Appendix}
\subsection{Proof of Lemma~\ref{lem3}} \label{app:a}
Here, we provide an explicit proof of Lemma~\ref{lem3} and details of the numerical simulation results.

For given $\rho = \sum^r_{i=1} p_i \proj{\psi_i}{\psi_i}$, let us define $T_0 = \sum^r_{i=1} t_i\proj{\psi_i}{\psi_i}$ with $t_i = \begin{cases} \log(\frac{p_i}{k}), & \tn{if}\;\;p_i \geq k; \\ 0, & \tn{if}\;\;p_i < k \end{cases}$ and $k=\frac{\varepsilon}{d^2}$. Then, the bound on the value of $\left|S\left(\rho\right)-f\left(T_0\right)\right|$ can be derived as:

\begin{widetext}
\begin{align*}
\left|S\left(\rho\right) - f\left(T_0\right)\right| 
&= \left|S\left(\rho\right)+ \T\left(\rho T_0\right) - \log\left(\T(e^{T_0})\right)\right| \\
& = \left|\sum^r_{i=1} p_i\log\left(\frac{1}{p_i}\right) + \sum_{p_i \geq k} p_i\log\left(\frac{p_i}{k}\right) - \log\left(\sum_{p_i<k} 1 + \sum_{p_i \geq k} \frac{p_i}{k}\right)\right| \\
& = \left|\sum_{p_i < k} p_i\log\left(\frac{1}{p_i}\right) + \sum_{p_i \geq k} p_i\log\left(\frac{1}{k}\right)- \log\left(\sum_{p_i<k} 1 + \sum_{p_i \geq k} \frac{p_i}{k}\right)\right| \\
& = \left|\sum_{p_i < k} p_i\log\left(\frac{1}{p_i}\right) + \sum_{p_i \geq k} p_i\log\left(\frac{1}{k}\right)-\log\left(\frac{1}{k}\right)+\log\left(\frac{1}{k}\right)-\log\left(\sum_{p_i<k} 1 + \sum_{p_i \geq k} \frac{p_i}{k}\right)\right| \\
& = \left|\sum_{p_i < k} p_i\log\left(\frac{k}{p_i}\right) - \log\left(\sum_{p_i<k} k + \sum_{p_i \geq k} p_i\right)\right| \\
& = \left|\sum_{p_i < k} p_i\log\left(\frac{1}{p_i}\right) + \sum_{p_i < k} p_i\log\left(\frac{1}{k}\right) + \log\left(1 + \sum_{p_i<k} \left(k-p_i\right)\right)\right| \\
& \leq 2dk\log\left(\frac{1}{k}\right) + dk = \frac{2\varepsilon^2\left(2\log d + \log\left(\frac{1}{\varepsilon}\right)\right)}{d} + \frac{\varepsilon}{d} \\
&< \varepsilon.
\end{align*}
\end{widetext}

That is, $\left|S(\rho)-f(T_0)\right| < \varepsilon$. Finally, $\T\left(T_0\right)$ is estimated as
\begin{align*}
\T\left(T_0\right) 
& = \sum_{p_i \geq k} \log\left(\frac{p_i}{k}\right) \leq \sum_{p_i \geq k} \log\left(\frac{1}{k}\right) \\
& \leq r\log\left(\frac{1}{k}\right) = 2r\log d + r\log\left(\frac{1}{\varepsilon}\right) \\
&= 2rn+r\log\left(\frac{1}{\varepsilon}\right).
\end{align*}
This implies that there exists a positive Hermitian matrix $T_0$ such that $\T(T_0) = O(rn+r\log\left(\frac{1}{\varepsilon}\right))$ and $\left|S\left(\rho\right)-f\left(T_0\right)\right| < \varepsilon$. $\square$

\subsection{Details on Numerical Simulations} \label{app:simul}
To support our observations, we explain the details of the numerical simulations for estimating the quantum mutual information, which can be expressed as the sum of von Neumann entropies as follows:
\begin{align*}
I\left(A:B\right) &= S(\rho^A) + S(\rho^B) - S(\rho^{AB}) \\
&= S\left(\rho^A\otimes\rho^B\right) - S(\rho^{AB}).
\end{align*}
To obtain quantum mutual information, we adopted an alternative and simple strategy. By exploiting QMINE (suggested in Sec. ~\ref{sec:MAIN}), we directly estimate $S(\rho^A\otimes\rho^B)$ and $S(\rho^{AB})$. That is, we address $S(\rho^A\otimes\rho^B)$ rather than estimating $S(\rho^{A})$ or $S(\rho^{B})$. This method reduces the number of resource copies required for simulations.

We used  four-qubit for this simulation and the results of our experiment are summarized in Table~\ref{tab:QMI}. To show that QMINE can estimate the quantum mutual information for various density matrices, we present the results of the estimation, where the rank of $\rho_{AB}$ is different.
\bigskip

\begin{table}[b]
\centering
\tabcolsep=0.05in
\begin{tabular}{c c c c} 
\hline\hline
{Rank } & {QMI} & {Estimation results} & {Error-rate (\%)} \\
\hline
1 & 1.8048002 & 1.7946120 & $0.565$ \\
2 & 1.7631968 & 1.7493981 & $0.783$ \\
4 & 1.3031208 & 1.2902124 & $0.991$ \\
8 & 0.8440226 & 0.8376048 & $0.760$ \\
16 & 0.4172888 & 0.4163618 & $0.222$ \\
\hline\hline
\end{tabular}
\caption{\label{tab:QMI} \textbf{Estimations of quantum mutual information using the QMINE method}}
\end{table}

%

\end{document}